\documentclass[runningheads,a4paper]{llncs}

\usepackage{amsmath}
\usepackage{amssymb}
\usepackage{theorem}
\usepackage{graphicx}
\usepackage{hyperref}
\usepackage[left=43mm,right=43mm,top=43mm,bottom=45mm]{geometry}

\newcommand{\Suff}{\textit{Suff}}
\newcommand{\Pref}{\textit{Pref}}
\newcommand{\Fact}{\textit{Fact}}

\newcommand{\AS}{\textit{ASF}}

\def\nats{\mathbb{N}}
\def\cd3#1{\textbf{\textsf{#1}}}
\def\sa#1{\cd3{#1}}

\renewcommand{\epsilon}{\varepsilon}
\renewcommand{\phi}{\varphi}

\begin{document}

\title{Words with the Maximum Number of Abelian Squares}

\author{
Gabriele Fici\inst{1} 
and Filippo Mignosi\inst{2}
}

\institute{
Dipartimento di Matematica e Informatica, Universit\`a di Palermo, Italy\\ \email{Gabriele.Fici@unipa.it}  
\and Dipartimento di Ingegneria e Scienze dell'Informazione e Matematica, Universit\`a dell'Aquila, Italy\\ \email{Filippo.Mignosi@di.univaq.it}
}

\maketitle

\begin{abstract}
An abelian square is the concatenation of two words that are anagrams of one another. A word of length $n$ can contain $\Theta(n^2)$ distinct factors that are abelian squares. We study infinite words such that the number of abelian square factors of length $n$ grows quadratically with $n$.
\end{abstract}

\keywords Abelian square, Thue-Morse word, Sturmian words, abelian-square rich word.

\section{Introduction}

A fundamental topic in Combinatorics on Words is the study of repetitions. A repetition in a word is a factor that is formed by the concatenation of two or more identical blocks. The simplest kind of repetition is a square, that is the concatenation of two copies of the same block, like $sciascia$.  A famous  conjecture of Fraenkel and Simpson \cite{FS98}  states that a word of length $n$ contains less than $n$  distinct square factors. Experiments strongly suggest that the conjecture is true, but a theoretical proof of the conjecture seems difficult. In \cite{FS98}, the authors proved a bound of $2n$. In \cite{I07}, Ilie improved this bound to $2n-\Theta(\log n)$, but the conjectured bound is still far away.

Among the different generalizations of the notion of repetition, a prominent one is that of an abelian repetition. An abelian repetition in a word is a factor that is formed by the concatenation of two or more blocks that have the same number of occurrences of each letter in the alphabet. Of course, the simplest kind of abelian repetition is an abelian square, that is therefore the concatenation of a word with an anagram of itself, like $viavai$. Abelian squares were considered in 1961 by Erd\"os \cite{Erdos1961221}, who conjectured that there exist infinite words avoiding abelian squares (this conjecture has later been proved to be true, and the smallest possible size of an alphabet for which it holds has been proved to be $4$ \cite{Ker92}).

We focus on the maximum number of abelian squares that a word can contain. Opposite to case of ordinary squares, a word of length $n$ can contain $\Theta(n^2)$ distinct abelian square factors (see \cite{Ry14}). Since the total number of factors in a word of length $n$ is quadratic in $n$, this means that there exist words in which a fixed proportion of all factors are abelian squares. So we turn our attention to infinite words, and we wonder whether there exist infinite words such that for every $n$ any factor of length $n$ contains, on average, a number of abelian squares that is quadratic in $n$. We call such an infinite word \emph{abelian-square rich}. Since a random binary word of length $n$ contains $\Theta(n\sqrt{n})$  distinct abelian square factors \cite{Ch14}, the existence of abelian-square rich words is not immediate. We also introduce \emph{uniformly abelian-square rich} words, that are infinite words such that for every $n$, every factor of length $n$ contains a quadratic number of abelian squares.

As a first result, we prove that the famous Thue-Morse word is uniformly abelian-square rich. Then we look at the class of Sturmian words, that are aperiodic infinite words with the lowest factor complexity. In this case, we prove that if a Sturmian word is $\beta$-power free for some $\beta\geq 2$ (that is, does not contain repetitions of order $\beta$ or higher), then it is uniformly abelian-square rich. 

\section{Notation and Background}

Let $\Sigma=\{a_1,a_2,\ldots,a_{\sigma}\}$ be an ordered $\sigma$-letter alphabet. Let $\Sigma^{*}$ stand for the free monoid generated by $\Sigma$, whose elements are called \emph{words} over $\Sigma$. The \emph{length} of a word $w$ is denoted by $|w|$. The \emph{empty word}, denoted by $\epsilon$, is the unique word of length zero and is the neutral element of $\Sigma^{*}$. We also define $\Sigma^{+}=\Sigma^{*}\setminus \{\epsilon\}$.
 
A \emph{prefix} (resp.~a \emph{suffix}) of a word $w$ is any word $u$ such that $w=uz$ (resp.~$w=zu$) for some word $z$. A \emph{factor} of $w$ is a prefix of a suffix (or, equivalently, a suffix of a prefix) of $w$.  . 
The set of prefixes, suffixes and factors of the word $w$ are denoted  by $\Pref(w)$, $\Suff(w)$ and $\Fact(w)$, respectively.
From the definitions, we have that $\epsilon$ is a prefix, a suffix and a factor of any word. 

For a word $w$ and a letter $a_i\in \Sigma$, we let $|w|_{a_i}$ denote the number of occurrences of $a_i$ in $w$. The \emph{Parikh vector} (sometimes called \emph{composition vector}) of a word $w$ over $\Sigma=\{a_1,a_2,\ldots,a_{\sigma}\}$ is the vector $P(w)=(|w|_{a_1},|w|_{a_2},\ldots,|w|_{a_{\sigma}})$. 
An \emph{abelian $k$-power} is a word of the form $v_1v_2\cdots v_k$ where all the $v_i$'s  have the same Parikh vector. An abelian $2$-power is called an \emph{abelian square}.

An \emph{infinite word} $w$ over $\Sigma$ is an infinite sequence of letters from $\Sigma$, that is, a function $w:\nats \mapsto \Sigma$. Given an infinite word $w$, the \emph{recurrence index} $R_w(n)$ of $w$ is the least integer $m$ (if any exists) such that every factor of $w$ of length $m$ contains all factors of $w$ of length $n$. If the recurrence index is defined for every $n$, the infinite word $w$ is called \emph{uniformly recurrent} and the function $R_w(n)$ the \emph{recurrence function} of $w$. A uniformly recurrent word  $w$ is called  \emph{linearly recurrent} if the ratio $R_w(n)/n$ is bounded. Given a linearly recurrent word $w$, the real number $r_w=\limsup_{n\to\infty}R_w(n)/n$ is called the \emph{recurrence quotient} of $w$. 

The \emph{factor complexity function} of an infinite word $w$ is the integer function $p_w(n)$ defined by $p_w(n)=|\Fact(w)\cap \Sigma^n|$. An infinite word $w$ has \emph{linear complexity} if $p_w(n)=O(n).$

A \emph{substitution} over the alphabet $\Sigma$ is a map $\tau:\Sigma \mapsto \Sigma^+$. Using the extension to words by concatenation, a substitution can be iterated. Note that for every substitution $\tau$ and every $n>0$, $\tau^n$ is again a substitution. Moreover, a substitution $\tau$ over $\Sigma$ can be naturally extended to a morphism from $\Sigma^*$ to $\Sigma^*$, since for every $u,v\in \Sigma^*$, one has $\tau(uv)=\tau(u)\tau(v)$, provided that one defines $\tau(\epsilon)=\epsilon$. A substitution $\tau$ is \emph{$k$-uniform} if there exists an integer $k\geq 1$ such that for all $a\in \Sigma$, $|\tau(a)|=k$. We say that a substitution is \emph{uniform} if it is $k$-uniform for some $k\geq 1$. A substitution $\tau$ is \emph{primitive} if there exists an integer $n\geq 1$ such that for every $a\in \Sigma$, $\tau^n(a)$ contains every letter of $\Sigma$ at least once. In this paper, we will only consider primitive substitutions such that $\tau(a_1)=a_1v$ for some non-empty word $v$. These substitutions always have a fixed point, which is the infinite word $w=\lim_{n\to\infty}\tau^n(a_1)$. Moreover, this fixed point is linearly recurrent (see for example \cite{DZ00}) and therefore has linear complexity.

\section{Abelian-square Rich Words}

Kociumaka et al.~\cite{Ry14} showed that a word of length $n$ can contain a number of distinct abelian square factors that is quadratic in $n$. We give here a proof of this fact for the sake of completeness.

\begin{proposition}
 A word of length $n$ can contain $\Theta(n^{2})$ distinct abelian square factors.
\end{proposition}

\begin{proof}
 Consider the word $w_n=a^{n}ba^{n}ba^{n}$, of length $3n+2$. For every $0\leq i,j\leq n$ such that $i+j+n$ is even, the factor $a^{i}ba^{n}ba^{j}$ of $w$ is an abelian square. Since the number of possible choices for the pair $(i,j)$ is quadratic in $n$, we are done. 
 \qed
\end{proof}

Motivated by the previous result, we wonder whether there exist infinite words such that all their factors contain a number of abelian squares that is quadratic in their length. But first, we relax this condition and consider words in which, for every sufficiently large $n$, a factor of length $n$ contains, on average, a number of distinct abelian square factors that is quadratic in $n$.

\begin{definition}
An infinite word $w$ is abelian-square rich if and only if there exists a positive constant $C$ such that for every $n$ sufficiently large one has 
\[\frac{1}{p_w(n)}\sum_{v\in \Fact(w)\cap\Sigma^n}{\{\mbox{\# abelian square factors of $v$}\}}\geq C n^2.
\]
\end{definition}

Notice that Christodoulakis et al. \cite{Ch14} proved that a binary word of length $n$ contains  $\Theta(n\sqrt{n})$  distinct abelian square factors on average, hence an infinite binary random word is almost surely not  abelian-square rich.

Given a finite or infinite word $w$, we let $\AS_w(n)$ denote the number of abelian square factors of $w$ of length $n$. Of course,  $\AS_w(n)=0$ if $n$ is odd, so this quantity is significant only for even values of $n$.

The following lemma is a consequence of the definition of linearly recurrent word.

\begin{lemma}\label{lem:lin}
 Let $w$ be a linearly recurrent word. If there exists a constant $C$ such that for every $n$ sufficiently large one has $\sum_{m\leq n}\AS_{w}(m)\geq Cn^2$, then $w$ is abelian-square rich. 
\end{lemma}

In an abelian-square rich word the average number of abelian squares in a factor is quadratic in the length of the factor. A stronger condition is that \emph{every} factor contains a quadratic number of abelian squares. We thus introduce uniformly abelian-square rich words.

\begin{definition}
An infinite word $w$ is uniformly abelian-square rich if and only if there exists a positive constant $C$ such that for every  $n$ sufficiently large one has 
\[\inf_{v\in \Fact(w)\cap\Sigma^n}{\{\mbox{\# abelian square factors of $v$}\}}\geq C n^2.
\]
\end{definition}

Clearly, if a word is uniformly abelian-square rich, then it is also abelian-square rich, but the converse is not always true. However, in the case of linearly recurrent words, the two definitions are equivalent, as shown in the next lemma.

\begin{lemma}\label{lem:un}
 If $w$ is abelian-square rich and linearly recurrent, then it is uniformly abelian-square rich.
\end{lemma}

\begin{proof}
Since $w$ is linearly recurrent, there exists a positive integer $K$ such that every factor of $w$ of length $Kn$ contains all the factors of $w$ of length $n$. Let $v$ be a factor of $w$ of length $n$ containing the largest number of abelian squares among the factors of $w$ of length $n$. Hence the number of abelian squares in $v$ is at least the average number of abelian squares in a factor of $w$ of length $n$. Since $w$ is abelian square rich, the number of abelian squares in $v$ is greater than or equal to $Cn^2$, for a positive constant $C$ and $n$ sufficiently large. Since $v$ is contained in any factor of $w$ of length $Kn$, the number of abelian squares in any factor of $w$ of length $Kn$ is greater than or equal to $Cn^2$, whence the statement follows.
 \qed
\end{proof}

The rest of this section is devoted to prove that the Thue-Morse word and the Sturmian words that do not contain arbitrarily large repetitions are uniformly abelian-square rich.

\subsection{The Thue-Morse Word}

Let 
\[
t=011010011001011010010110\cdots
\]
be the Thue-Morse word, i.e., the fixed point of the uniform substitution $\mu:0\mapsto 01,1\mapsto 10$. 
For every $n\geq 4$, the factors of length $n$ of $t$ belong to two disjoint sets: those that start only at even positions in $t$, and those that start only at odd positions in $t$. 
This is a consequence of the fact that $t$ is overlap-free,  hence $0101$ cannot be preceded by $1$ nor followed by $0$, and that $00$ and $11$ are not images of letters, so they cannot appear at even positions.

Let $p(n)$ be the factor complexity function of $t$. It is known \cite[Proposition 4.3]{Br89}, that for every $n\geq 1$ one has $p(2n)=p(n)+p(n+1)$ and $p(2n+1)=2p(n+1)$.

The next lemma (proved in \cite{CaFiScZa15}) shows that the Thue-Morse word has the property that for every length there are at least one third of the factors that begin and end with the same letter, and at least one third of the factors that begin and end with different letters. We define  $f_{aa}(n)$ (resp.~$f_{ab}(n)$) as the number of factors of $t$ of length $n$ that begin and end with the same letter (resp.~with different letters).

\begin{lemma}[\cite{CaFiScZa15}]\label{lem:third}
 For every $n\geq 2$, one has $f_{aa}(n)\geq  p(n)/3$ and $f_{ab}(n)\geq  p(n)/3$.
\end{lemma}

Since $p(n)\geq 3(n-1)$ for every $n$ \cite[Corollary 4.5]{DelVa89}, we get the following result.

\begin{corollary}\label{cor:TM}
  For every $n\geq 2$, one has $f_{aa}(n)\geq  n-1$ and $f_{ab}(n)\geq  n-1$.
\end{corollary}

\begin{proposition}
 The Thue-Morse word $t$ is uniformly abelian-square rich.
\end{proposition}

\begin{proof}
 Let $u$ be a factor of length $n>1$ of $t$ that begins and ends with the same letter. Since the image of any even-length word under $\mu$ is an abelian square, we have that $\mu^2(u)$ is an abelian square factor of $t$ of length $4n$ that begins and ends with the same letter. Moreover, the word obtained from $\mu^2(u)$ by removing the first and the last letter  is an abelian square factor of $t$ of length $4n-2$. So, by Corollary \ref{cor:TM}, $t$ contains at least $n-1$ abelian square factors of length $4n$ and at least $n-1$ abelian square factors of length $4n-2$. This implies that for every even $n$ the number of abelian square factors of $t$ of length $n$ is linear in $n$. Hence, for every $n$ the number of abelian square factors of $t$ of length at most $n$ is quadratic in $n$. The statement then follows from Lemmas \ref{lem:lin} and \ref{lem:un}.
 \qed
\end{proof}

\subsection{Sturmian Words}

In this section we fix the alphabet $\Sigma=\{\sa{a,b}\}$.

Recall that a (finite or infinite) word $w$ over $\Sigma$ is \emph{balanced} if and only if for any $u,v$ factors of $w$ of the same length, one has $||u|_{\sa{a}}-|v|_{\sa{a}}|\le 1$.  

We start with a simple lemma.

\begin{lemma}\label{lem:bal}
 Let $w$ be a finite balanced word over $\Sigma$. Then for any $k>0$,  $P(w)=(0,0)\mod k$ if and only if $w$ is an abelian $k$-power.
\end{lemma}

\begin{proof}
Let $w$ be balanced and $P(w)=(ks,kt)$, for a positive integer $k$ and some  $s,t\geq 0$. Then we can write $w=v_1v_2\cdots v_k$ where each $v_i$ has length $s+t$. Now, each $v_i$ must have Parikh vector equal to $(s,t)$ otherwise $w$ would not be balanced, whence the only if part of the statement follows. The if part is straightforward.
\qed
\end{proof}

A binary infinite word is \emph{Sturmian} if and only if it is balanced and aperiodic. Sturmian words are precisely the infinite words having $n+1$ distinct factors of length $n$ for every $n\geq 0$. There is a lot of other equivalent definitions of Sturmian words. A classical reference on Sturmian words is \cite[Chapter 2]{LothAlg}. Let us recall here the definition of Sturmian words as codings of a rotation.

We fix the torus $I=\mathbb{R}/\mathbb{Z}=[0,1)$. Given $\alpha,\beta$ in $I$, if $\alpha > \beta$, we use the notation $[\alpha,\beta)$ for the interval $[\alpha,1)\cup [0,\beta)$. Recall that given a real number $\alpha$,  $\lfloor \alpha \rfloor$ is the greatest integer smaller than or equal to $\alpha$, $\lceil \alpha \rceil$ is the smallest integer greater than or equal to $\alpha$, and $\{\alpha\}=\alpha-\lfloor \alpha \rfloor$ is the fractional part of  $\alpha$. Notice that $\{-\alpha\}= 1-\{\alpha\}$. 

Let $\alpha\in I$ be irrational, and $\rho\in I$. 
The Sturmian word $s_{\alpha,\rho}$ (resp.~$s'_{\alpha,\rho}$) of  \emph{angle} $\alpha$ and \emph{initial point} $\rho$ is the infinite word $s_{\alpha,\rho}=a_{0}a_{1}a_{2}\cdots$ defined by
$$a_{n} =
\left\{
	\begin{array}{ll}
		\sa{b}  & \mbox{if } \{ \rho + n\alpha \}\in I_{\sa{b}},\\
		\sa{a}  & \mbox{if } \{ \rho + n\alpha \}\in I_{\sa{a}}, 
	\end{array}
\right.$$ 
where $I_{\sa{b}}=[0,1-\alpha)$ and $I_{\sa{a}}=[1-\alpha,1)$   (resp.~$I_{\sa{b}}=(0,1-\alpha]$ and $I_{\sa{a}}=(1-\alpha,1]$).

In other words, take the unitary circle and consider a point initially in position $\rho$. Then start rotating this point on the circle (clockwise) of an angle $\alpha$, $2\alpha$, $3\alpha$, etc. For each rotation, take the letter $\sa{a}$ or $\sa{b}$ associated with the interval within which the point falls. The infinite sequence obtained in this way is the Sturmian word $s_{\alpha,\rho}$ (or $s'_{\alpha,\rho}$, depending on the choice of the two intervals). See \figurename~\ref{Fig:gab1} for an illustration.

\begin{figure}[!ht]
\centering
\includegraphics[scale=0.3]{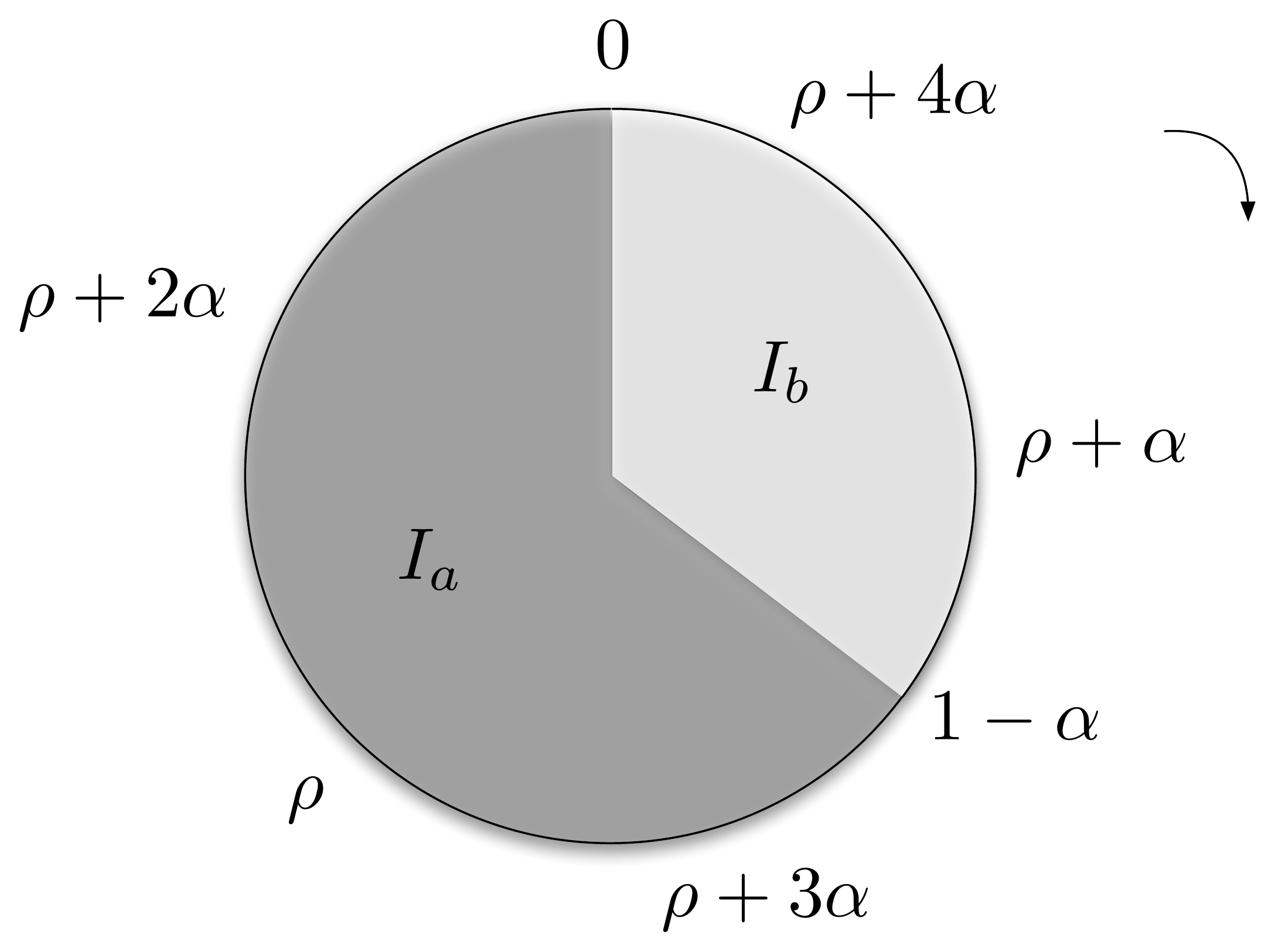} 
\caption{The rotation of angle $\alpha=\phi-1\approx0.618$ and initial point $\rho=\alpha$ generating the Fibonacci word $F=s_{\phi-1,\phi-1}=\sa{abaababaabaabab}\cdots$. \label{Fig:gab1}}
\end{figure}

For example, if $\phi=(1+\sqrt 5 )/2\approx 1.618$ is the golden ratio,  the Sturmian word $$F=s_{\phi-1,\phi-1}=\sa{abaababaabaababaababaabaababaabaab}\cdots$$ is called the \emph{Fibonacci word}: 

A Sturmian word for which $\rho=\alpha$, like the Fibonacci word, is called \emph{characteristic}. Note that for every $\alpha$ one has $s_{\alpha,0}=\sa{b}s_{\alpha,\alpha}$ and $s'_{\alpha,0}=\sa{a}s_{\alpha,\alpha}$.

An equivalent way to see the coding of a rotation consists in fixing the point and rotating the intervals.  In this representation, the interval $I_{\sa{b}}=I_{\sa{b}}^{0}$ is rotated at each step, so that after $i$ rotations it is transformed into the interval $I_{\sa{b}}^{-i}=[\{-i\alpha\},\{-(i+1)\alpha\})$, while $I_{\sa{a}}^{-i}=I\setminus I_{\sa{b}}^{-i}$. 

This representation is convenient since one can read within it not only a Sturmian word but also any of its factors. More precisely, for every positive integer $n$, the factor of length $n$ of $s_{\alpha,\rho}$ starting at position $j\geq 0$ is determined by the value of $\{\rho+j\alpha\}$ only. Indeed, for every $j$ and $i$, we have:
 $$a_{j+i} = \left\{ \begin{array}{lllll}
\sa{b} & \mbox{if $\{\rho + j\alpha\}\in I_{\sa{b}}^{-i}$;}\\
\sa{a} & \mbox{if $\{\rho + j\alpha\}\in I_{\sa{a}}^{-i}$.}
\end{array} \right.$$
As a consequence, we have that given a Sturmian word $s_{\alpha,\rho}$ and  a positive integer $n$, the $n+1$ different factors of $s_{\alpha,\rho}$ of length $n$ are completely determined by the intervals $I_{\sa{b}}^{0}, I_{\sa{b}}^{-1},\ldots, I_{\sa{b}}^{-(n-1)}$, that is, only by the points $\{-i\alpha\}$, $0\leq i< n$. In particular, they do not depend on $\rho$, so that the set of factors of $s_{\alpha,\rho}$ is the same as the set of factors of $s_{\alpha,\rho'}$ for any $\rho$ and $\rho'$.
Hence, from now on, we let $s_{\alpha}$ denote any Sturmian word of angle $\alpha$.

If we arrange the $n+2$ points $0,1,\{-\alpha\},\{-2\alpha\},\ldots,\{-n \alpha\}$ in increasing order, we determine a partition of $I$ in $n+1$ subintervals, $L_0(n),L_{1}(n),\ldots,L_{n}(n)$. Each of these subintervals is in bijection with a different factor of length $n$ of any Sturmian word of angle $\alpha$ (see \figurename~\ref{Fig:gab3}). 

\begin{figure}[!ht]
\centering
\includegraphics[scale=0.5]{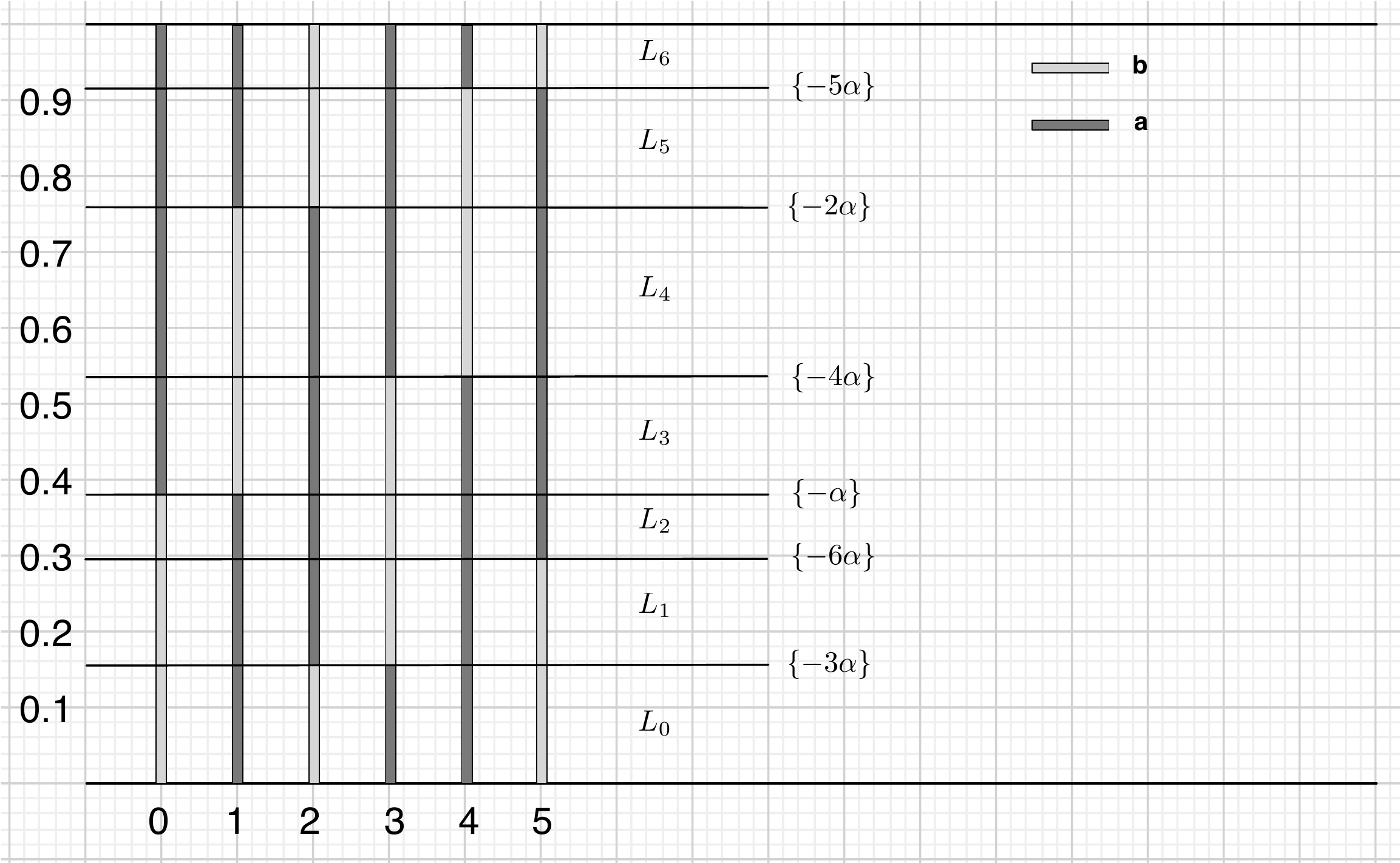} 
\caption{%
The points $0$, $1$ and  $\{-\alpha\}$, $\{-2\alpha\}$, $\{-3\alpha\}$, $\{-4\alpha\}$, $\{-5\alpha\}$, $\{-6\alpha\}$, 
arranged in increasing order, define the intervals $L_{0}(6)\approx[0,0.146)$, $L_{1}(6)\approx[0.146,0.292)$, $L_{2}(6)\approx[0.292,0.382)$, $L_{3}(6)\approx[0.382,0.528)$, $L_{4}(6)\approx[0.528,0.764)$, $L_{5}(6)\approx[0.764,0.910)$, $L_{6}(6)\approx[0.910,1)$. Each interval is associated with one of the factors of length $6$ of the Fibonacci word, respectively $\sa{babaab},\sa{baabab},\sa{baabaa},\sa{ababaa},\sa{abaaba},\sa{aababa},\sa{aabaab}$. \label{Fig:gab3}}
\end{figure}

Recall that a factor of length $n$ of a Sturmian word  $s_{\alpha}$ has a Parikh vector equal either to $(\lfloor n\alpha \rfloor , n-\lfloor n\alpha \rfloor )$ (in which case it is called \emph{light}) or to $(\lceil n\alpha \rceil , n-\lceil n\alpha \rceil) $ (in which case it is called \emph{heavy}).
The following proposition relates the intervals $L_{i}(n)$ to the Parikh vectors of the factors of length $n$ (see  \cite{dlt} and  \cite{Rigo13}).

\begin{proposition}\label{pro:main}
Let $s_{\alpha}$ be a Sturmian word of angle $\alpha$, and $n$ a positive integer.
Let $t_{i}$ be the factor of length $n$ associated with the interval $L_{i}(n)$. Then $t_{i}$ is
  heavy if  $L_{i}(n)\subset [\{-n\alpha\},1)$, while it is light if  $L_{i}(n)\subset [0,\{-n\alpha\})$.
\end{proposition}

\begin{example}
 Let $\alpha=\phi-1\approx 0.618$ and $n=6$. We have $6\alpha\approx 3.708$, so that $\{-6\alpha\}\approx 0.292$. The reader can see in \figurename~\ref{Fig:gab3} that the factors of length $6$ 
corresponding to intervals above (resp.~below) $\{-6\alpha\}\approx 0.292$ all have Parikh vector $(4,2)$ (resp.~$(3,3)$). That is, the intervals $L_{0}$ and $L_{1}$ are associated with light factors (\sa{babaab}, \sa{baabab}), while the intervals $L_{2}$ to $L_{6}$ are associated with  heavy factors (\sa{baabaa}, \sa{ababaa}, \sa{abaaba}, \sa{aababa}, \sa{aabaab}). 
\end{example}

Observe that, by Lemma \ref{lem:bal}, every factor of a Sturmian word having even length and containing an even number of $\sa{a}$'s (or, equivalently, of $\sa{b}$'s) is an abelian square. The following proposition relates the abelian square factors of a Sturmian word of angle $\alpha$  with the arithmetic properties of $\alpha$.

\begin{proposition}\label{prop:hl}
Let $s_{\alpha}$ be a Sturmian word of angle $\alpha$, and $n$ a positive even integer.
Let $t_{i}$ be the factor of length $n$ associated with the interval $L_{i}(n)$. Then $t_{i}$ is
an abelian square if and only if $L_{i}(n)\subset [\{-n\alpha\},1)$ if $\lfloor n \alpha \rfloor$ is even, or $L_{i}(n)\subset [0,\{-n\alpha\})$ if $\lfloor n \alpha \rfloor$ is odd.
\end{proposition}

\begin{proof}
By Proposition \ref{pro:main}, $t_i$ is heavy if and only if $L_{i}(n)\subset [\{-n\alpha\},1)$, while it is light if and only if $L_{i}(n)\subset [0,\{-n\alpha\})$. If $\lfloor n \alpha \rfloor$ is even, then every light factor of length $n$ contains an even number of $\sa{a}$'s and hence is an abelian square, while if $\lfloor n \alpha \rfloor$ is odd, then every heavy factor of length $n$ contains an even number of $\sa{a}$'s and hence is an abelian square, whence the statement follows.
 \qed
\end{proof}

Recall that given a finite or infinite word $w$,  $\AS_w(n)$ denotes the number of abelian square factors of $w$ of length $n$. 

\begin{corollary}\label{cor:formula}
 Let $s_{\alpha}$ be a Sturmian word of angle $\alpha$.
For every positive even $n$, let $I_{n}=\{\{-i \alpha\}\mid 1\le i \le n\}$. Then
\[
\AS_{s_{\alpha}}(n)=
\begin{cases}
 \# \{x\in I_{n} \mid x \leq \{-n \alpha\}\} \mbox{ if $\lfloor n \alpha \rfloor$ is even; }\\
 \# \{x\in I_{n} \mid x \geq \{-n \alpha\}\} \mbox{ if $\lfloor n \alpha \rfloor$ is odd. }\\ 
\end{cases}
\]
\end{corollary}

\begin{example}
 The factors of length $6$ of the Fibonacci word $F$ are, lexicographically ordered:
$\sa{aabaab, aababa, abaaba, ababaa, baabaa}$ (heavy factors), $\sa{baabab, babaab}$ (light factors).
The light factors,  whose number of $\sa{a}$'s is $\lfloor 6  \alpha \rfloor  =3$,  are not abelian squares; the heavy factors,  whose number of $\sa{a}$'s is $\lceil 6 \alpha \rceil=4$,  are all abelian squares.

We have $I_{6}=\{0.382, 0.764, 0.146, 0.528, 0.910, 0.292\}$ (values are approximated) and $6\alpha\simeq 3.708$, so $\lfloor 6\alpha \rfloor$ is odd. Thus, there are $5$ elements in $I_{6}$ that are $\geq \{-6 \alpha\}$, so by Corollary \ref{cor:formula} there are $5$ abelian square factors of length $6$.

The factors of length $8$ of the Fibonacci word are, lexicographically ordered:
$\sa{aabaabab}$, $\sa{aababaab}$, $\sa{abaabaab}$, $\sa{abaababa}$, $\sa{ababaaba}$, 
$\sa{baabaaba}$, $\sa{baababaa}$, $\sa{babaabaa}$ (heavy factors), $\sa{babaabab}$ (light factor).
The light factor,  whose number of $\sa{a}$'s is $\lfloor 8  \alpha \rfloor  =4$,  is an abelian square; the heavy factors,  whose number of $\sa{a}$'s is $\lceil 8 \alpha \rceil=5$,  are not abelian squares.
We have $I_{8}=\{0.382, 0.764, 0.146, 0.528, 0.910, 0.292, 0.674, 0.056\}$ (values are approximated) and $8\alpha\simeq 4.944$, so $\lfloor 8\alpha \rfloor$ is even. Thus, there is only one element in $I_{8}$ that is $\leq \{8 \alpha\}$, so by Corollary \ref{cor:formula} there is only one abelian square factor of length $8$.

In Table \ref{tab:i} we report the first values of the sequence $\AS_F(n)$ for the Fibonacci word $F$.
\end{example}

\begin{table}[bt]
\centering  
\begin{small}
\begin{raggedright}
\begin{tabular}{c *{30}{@{\hspace{2.5mm}}c}}
$n$\hspace{2mm} &0  & 2  & 4  & 6  & 8 & 10  & 12  & 14  & 16  & 18  & 20 & 22 & 24 & 26 & 28 & 30 & 32 & 34 & 36
\\
\hline \\
$\AS_F(n)$\hspace{2mm} &0 & 1& 3& 5& 1& 9& 5& 5& 15& 3& 13& 13& 5& 25& 9& 15 & 25 & 21 & 27
\\
\hline \rule[0pt]{0pt}{12pt}
\end{tabular}
\end{raggedright}\caption{\label{tab:i} The first values of the sequence $\AS_{F}(n)$ of the number of abelian square factors of length $n$ in the Fibonacci word $F=s_{\phi-1,\phi-1}$.}
\end{small}
\end{table}

Recall that every irrational number $\alpha$ can be uniquely written as a (simple) continued fraction as follows:
\begin{equation}\label{cf}
 \alpha=a_{0}+\frac{1}{a_{1}+\frac{1}{a_{2}+\ldots}}
\end{equation}
where $a_{0}=\lfloor \alpha \rfloor$, and the infinite sequence $(a_{i})_{i\geq 0}$ is called the sequence of partial quotients of $\alpha$. The continued fraction expansion of $\alpha$ is usually denoted by its sequence of partial quotients as follows: $\alpha=[a_{0};a_{1},a_{2},\ldots ]$, and each its finite truncation $[a_{0};a_{1},a_{2},\ldots,a_{k}]$ is a rational number $n_{k}/m_{k}$ called the $k$th convergent to $\alpha$. We say that an irrational $\alpha=[a_{0};a_{1},a_{2},\ldots ]$ has bounded partial quotients if and only if the sequence $(a_{i})_{i\geq 0}$ is bounded.  

The development in  continued fraction of $\alpha$ is deeply related to the exponent of the factors of the Sturmian word $s_{\alpha}$. Recall that an infinite word $w$ is said to be $\beta$-power free, for some $\beta\geq 2$, if for every factor $v$ of $w$, the ratio between the length of $v$ and its minimal period is smaller than $\beta$.
The second author \cite{Mi89} proved that a Sturmian word of angle $\alpha$ is $\beta$-power free for some $\beta\geq 2$ if and only if $\alpha$ has bounded partial quotients.

Since the golden ratio $\phi$ is defined by the equation $\phi=1+1/\phi$, we have from Equation \ref{cf} that $\phi=[1;1,1,1,1,\ldots]$ and therefore $\phi-1=[0;1,1,1,1,\ldots]$, so the Fibonacci word is an example of $\beta$-power free Sturmian word (actually, it is $(2+\phi)$-power free \cite{MignosiPirillo}).

We are now proving that if $\alpha$ has bounded partial quotients, then the Sturmian word  $s_{\alpha}$ is abelian-square rich. For this, we will use a result on the discrepancy of uniformly distributed modulo $1$ sequences from \cite{KuNi}. To the best of our knowledge, this is the first application of this result to the theory of Sturmian words, and we think that the correspondence we are now showing might be useful for deriving other results on Sturmian words.

Let $\omega=(x_n)_{n\geq 0}$ be a given sequence of real numbers. For a positive integer $N$ and a subset $E$ of the torus $I$, we define $A(E;N;\omega)$ as the number of terms $x_n$, $0\leq n\leq N$, for which $\{x_n\}\in E$. If there is no risk of confusion, we will write $A(E;N)$ instead of $A(E;N;\omega)$.

\begin{definition}
 The sequence $\omega=(x_n)_{n\geq 0}$ of real numbers is said to be uniformly distributed modulo $1$ if and only if for every pair $a,b$ of real numbers with $0\leq a<b\leq 1$ we have
 \[
 \lim_{N\to \infty} \frac{A([a,b);N;\omega)}{N}=b-a.
 \]
 \end{definition}
 
 \begin{definition}
 Let $x_0,x_1,\ldots, x_N$ be a finite sequence of real numbers. The number
 \[
 D_N=D_N(x_0,x_1,\ldots, x_N)=\sup_{0\leq \gamma<\delta\leq 1}\left| \frac{A([\gamma,\delta);N)}{N}-(\delta-\gamma)\right|
 \]
 is called the discrepancy of the given sequence. For an infinite sequence $\omega$ of real numbers the discrepancy $D_{N}(\omega)$ is the discrepancy of the initial segment formed by the first $N+1$ terms of $\omega$.
\end{definition}

The two previous definitions are related by the following result.

\begin{theorem}[\cite{KuNi}]
 The sequence $\omega$ is uniformly distributed modulo $1$ if and only if $\lim_{N\to \infty}D_N(\omega)=0$.
\end{theorem}

An important class of  uniformly distributed modulo $1$ sequences is given by the sequence $(n\alpha)_{n\geq 0}$ with $\alpha$ a given irrational number and $n\in \nats$. The discrepancy of the sequence $(n\alpha)$ will depend on the finer arithmetical properties of $\alpha$. In particular, we have the following theorem, stating that if $\alpha$ has bounded partial quotients, then its discrepancy has the least order of magnitude possible.

\begin{theorem}[\cite{KuNi}]\label{theor:KN2}
 Suppose the irrational $\alpha=[a_0;a_1,\ldots]$ has  partial quotients bounded by $K$. Then the discrepancy $D_N(\omega)$ of $\omega=(n\alpha)$ satisfies $ND_N(\omega)=O(\log N)$. More exactly, we have
\begin{equation}
 ND_N(\omega)\leq 3+\left(\frac{1}{\log \phi}+\frac{K}{\log(K+1)}\right)\log N.
\end{equation}
\end{theorem}

We are now using previous definitions and results to prove that $\beta$-power free Sturmian words are abelian-square rich.

\begin{theorem}\label{theor:lin}
 Let $s_{\alpha}$ be a Sturmian word of angle $\alpha$ such that $\alpha$ has bounded partial quotients. Then there exists a positive constant $C$ such that for every $n$ sufficiently large one has $\sum_{m\leq n}\AS_{s_{\alpha}}(m)\geq Cn^2$.
\end{theorem}

\begin{proof}

For every even $n$, let $I'_{n}=\{\{i \alpha\}\mid 1\le i \le n\}$. By Corollary \ref{cor:formula} and basic arithmetical properties of the fractional part, we have:
\[
\AS_{s_{\alpha}}(n)=
\begin{cases}
 \# \{x\in I'_{n} \mid x \geq \{n \alpha\}\} \mbox{ if $\lfloor n \alpha \rfloor$ is even; }\\
 \# \{x\in I'_{n} \mid x \leq \{n \alpha\}\} \mbox{ if $\lfloor n \alpha \rfloor$ is odd. }\\ 
\end{cases}
\]
So:
\begin{eqnarray}
&& \sum_{m\leq n}\AS_{s_{\alpha}}(m)  \\
&& \geq \sum_{m\leq n} \# \{ \{i \alpha\}\mid\{i \alpha\}\leq 1/2, i\leq m, \mbox{ and  }
\{m \alpha\}\leq 1/2,  \lfloor m\alpha\rfloor \mbox{  even }\} \\
&& \geq  \sum_{m\leq n} \# \{ \{i \alpha/2\}\mid\{i \alpha/2\}\in[1/4,1/2), i\leq m, \mbox{ and }
\{m \alpha/2\}\leq 1/4 \}   \\
 && \geq \hspace{-3mm} \sum_{n/2\leq m\leq n} \# \{ \{i \alpha/2\}\mid\{i \alpha/2\}\in[1/4,1/2), i\leq n/2, \mbox{ and }
\{m \alpha/2\}\leq 1/4 \} \\
&& = \# \{ \{i \alpha/2\}\mid\{i \alpha/2\}\in[1/4,1/2), i\leq n/2\}\times \hspace{-4mm} \sum_{n/2\leq m\leq n}  \{ m \mid \{m \alpha/2\}\leq 1/4 \}
\end{eqnarray}
where: $(4)$ follows from $(3)$ by Corollary \ref{cor:formula};  $(5)$ follows from $(4)$ because $\{m \alpha/2\}\leq 1/4$ implies $\{m \alpha\}\leq 1/2$ and $\lfloor m\alpha\rfloor$ is even if and only if $\{m \alpha/2\}\leq 1/2$;  $(6)$ follows from $(5)$ is obvious; finally $(7)$ follows from $(6)$ because the cardinality of the first set is independent from the sum.

Now, $\alpha/2$ has bounded partial quotients (since $\alpha$ has) and we can apply Theorem \ref{theor:KN2} to evaluate the two factors of $(7)$. So we have:
\begin{eqnarray*}
&& \# \{ \{i \alpha/2\}\mid\{i \alpha/2\}\in[1/4,1/2), i\leq n/2\} \\ 
&& = A([1/4,1/2);n/2;(n\alpha/2)) \\\
&& \geq (1/2-1/4)n/2-C_1\log n \\
&& = n/8-C_1\log n,
\end{eqnarray*}
for $n$ sufficiently large and a positive constant $C_1$. We also have:
\begin{eqnarray*}
&& \hspace{-3mm} \sum_{n/2\leq m\leq n} \{ m \mid \{m \alpha/2\}\leq 1/4 \} \\
&& =A([0,1/4);n;(n\alpha/2))-A([0,1/4);n/2;(n\alpha/2)) \\
&&  \geq n/4-C_2\log n - n/8 - C_3\log n \\
&& = n/8 - C_4 \log n,
\end{eqnarray*}
for $n$ sufficiently large and positive constants $C_2,C_3,C_4$. The product of the two factors of $(6)$ is therefore greater than a constant times $n^2$, as required.
 \qed
\end{proof}

The recurrence quotient $r_{\alpha}$ of a Sturmian word of angle $\alpha=[0;a_1,a_2,\ldots]$ such that $\alpha$ has bounded partial quotients verifies $2+r_{\alpha}<\limsup a_i <3+r_{\alpha}$ \cite[Prop.~5]{Ca99}. Moreover, Durand \cite{Du03} proved that a Sturmian word of angle $\alpha$ is linearly recurrent if and only if $\alpha$ has bounded partial quotients. Thus, we have the following:

\begin{corollary}
 Let $s_{\alpha}$ be a Sturmian word of angle $\alpha$. If $s_{\alpha}$ is $\beta$-power free, then $s_{\alpha}$ is uniformly abelian-square rich.
\end{corollary}

\begin{proof}
We known that $s_{\alpha}$ is $\beta$-power free for some  $\beta\geq 2$ if and only if $\alpha$ has bounded partial quotients if and only if $s_{\alpha}$ is linearly recurrent. The statement then follows from Theorem \ref{theor:lin} and Lemmas \ref{lem:lin} and \ref{lem:un}.
\qed
\end{proof}
 
\section{Conclusions and future work}

We proved that the Thue-Morse is uniformly abelian-square rich. We think that the technique we used for the proof can be generalized to some extent, and could be used, for example,  to prove that a class of fixed points of uniform substitutions are uniformly abelian-square rich.

We also proved that Sturmian words that are $\beta$-power free for some $\beta\geq 2$ are uniformly abelian-square rich. The proof we gave is based on a classical result on the discrepancy of the uniformly distributed modulo $1$ sequence $(n\alpha)_{n\geq 0}$, where $\alpha$ is the slope of the Sturmian word. To the best of our knowledge, this is the first application of this result to the theory of Sturmian words, and we think that the correspondence we have shown might be useful for deriving other results on Sturmian words.

The natural question that then arises is whether the hypothesis of power freeness is necessary for a Sturmian word being (uniformly) abelian-square rich. 
We leave open the question to determine whether $s_{\alpha}$ is not uniformly abelian-square rich nor abelian-square rich in the case when $\alpha$ has unbounded partial quotients. 

We mostly investigated binary words in this paper. We conjecture that binary words have the largest number of abelian square factors. More precisely, we propose the following conjecture.

\begin{conjecture}
 If a word of length $n$ contains $k$ many distinct abelian square factors, then there exists a binary word of length $n$ containing at least $k$ many distinct abelian square factors.
\end{conjecture}

A slightly different point of view from the one we considered in this paper consists in identifying two abelian squares if they have the same Parikh vector. Two abelian squares are therefore called \emph{inequivalent} if they have different Parikh vectors \cite{FSP97}. Sturmian words only have a linear number of inequivalent abelian squares. Nevertheless, a word of length $n$ can contain $\Theta(n\sqrt{n})$ inequivalent abelian squares \cite{warsaw}. Computations support the following conjecture:

\begin{conjecture}[see \cite{Ry14}]
 A word of length $n$ contains $O(n\sqrt{n})$ inequivalent abelian squares.
\end{conjecture}

\section{Acknowledgements}

The authors acknowledge the support of the PRIN 2010/2011 project ``Automi e Linguaggi Formali: Aspetti Matematici e Applicativi'' of the Italian Ministry of Education (MIUR).

\bibliographystyle{abbrv}
\bibliography{ref}

\begin{thebibliography}{10}

\bibitem{Br89}
S.~Brlek.
\newblock Enumeration of factors in the {T}hue-{M}orse word.
\newblock {\em Discr. Appl. Math.}, 24(1-3):83--96, 1989.

\bibitem{Ca99}
J.~Cassaigne.
\newblock Limit {V}alues of the {R}ecurrence {Q}uotient of {S}turmian
  {S}equences.
\newblock {\em Theoret. Comput. Sci.}, 218(1):3--12, 1999.

\bibitem{CaFiScZa15}
J.~Cassaigne, G.~Fici, M.~Sciortino, and L.~Zamboni.
\newblock Cyclic {C}omplexity of {W}ords.
\newblock {\em Submitted}, 2015.
\newblock Available at \url{http://arxiv.org/abs/1402.5843}.

\bibitem{Ch14}
M.~Christodoulakis, M.~Christou, M.~Crochemore, and C.~S. Iliopoulos.
\newblock On the average number of regularities in a word.
\newblock {\em Theoret. Comput. Sci.}, 525:3--9, 2014.

\bibitem{DZ00}
D.~Damanik and D.~Zare.
\newblock Palindrome complexity bounds for primitive substitution sequences.
\newblock {\em Discrete Mathematics}, 222(1--3):259--267, 2000.

\bibitem{DelVa89}
A.~de~Luca and S.~Varricchio.
\newblock Some combinatorial properties of the {T}hue-{M}orse sequence and a
  problem in semigroups.
\newblock {\em Theoret. Comput. Sci.}, 63(3):333--348, 1989.

\bibitem{Du03}
F.~Durand.
\newblock Corrigendum and addendum to: ``{L}inearly recurrent subshifts have a
  finite number of non-periodic subshift factors'' [{E}rgodic {T}heory {D}ynam.
  {S}ystems 20 (2000), no. 4, 1061--1078].
\newblock {\em Ergodic Theory Dynam. Systems}, 23(2):663--669, 2003.

\bibitem{Erdos1961221}
P.~Erd\"os.
\newblock Some unsolved problems.
\newblock {\em Magyar Tud. Akad. Mat. Kutato. Int. Kozl.}, 6:221--254, 1961.

\bibitem{dlt}
G.~Fici, A.~Langiu, T.~Lecroq, A.~Lefebvre, F.~Mignosi, and E.~Prieur-Gaston.
\newblock Abelian {R}epetitions in {S}turmian {W}ords.
\newblock In {\em Developments in Language Theory. Proceedings}, volume 7907 of
  {\em Lecture Notes in Computer Science}, pages 227--238. Springer, 2013.

\bibitem{FS98}
A.~S. Fraenkel and J.~Simpson.
\newblock How many squares can a string contain?
\newblock {\em Journal of Combinatorial Theory, Series A}, 82(1):112--120,
  1998.

\bibitem{FSP97}
A.~S. Fraenkel, J.~Simpson, and M.~Paterson.
\newblock On weak circular squares in binary words.
\newblock In {\em Combinatorial Pattern Matching. Proceedings}, volume 1264 of
  {\em Lecture Notes in Computer Science}, pages 76--82. Springer, 1997.

\bibitem{I07}
L.~Ilie.
\newblock A note on the number of squares in a word.
\newblock {\em Theoretical Computer Science}, 380(3):373--376, 2007.

\bibitem{Ker92}
V.~Ker\"{a}nen.
\newblock Abelian squares are avoidable on 4 letters.
\newblock In {\em Proceedings of the 19th International Colloquium on Automata,
  Languages and Programming}, volume 623 of {\em Lecture Notes in Comput.
  Sci.}, pages 41--52. Springer-Verlag, 1992.

\bibitem{Ry14}
T.~Kociumaka, J.~Radoszewski, W.~Rytter, and T.~Walen.
\newblock Maximum number of distinct and nonequivalent nonstandard squares in a
  word.
\newblock In {\em Developments in Language Theory. Proceedings}, volume 8633 of
  {\em Lecture Notes in Computer Science}, pages 215--226. Springer, 2014.

\bibitem{warsaw}
T.~Kociumaka, J.~Radoszewski, W.~Rytter, and T.~Walen.
\newblock Personal communication, 2015.

\bibitem{KuNi}
L.~Kuipers and H.~Niederreiter.
\newblock {\em Uniform distribution of sequences}.
\newblock John Wiley \& Sons, New York, NY, 1974.

\bibitem{LothAlg}
M.~Lothaire.
\newblock {\em Algebraic Combinatorics on {W}ords}.
\newblock Cambridge University Press, Cambridge, U.K., 2002.

\bibitem{Mi89}
F.~Mignosi.
\newblock Infinite {W}ords with {L}inear {S}ubword {C}omplexity.
\newblock {\em Theoret. Comput. Sci.}, 65(2):221--242, 1989.

\bibitem{MignosiPirillo}
F.~Mignosi and G.~Pirillo.
\newblock Repetitions in the {F}ibonacci infinite word.
\newblock {\em RAIRO Theor. Inform. Appl.}, 26:199--204, 1992.

\bibitem{Rigo13}
M.~Rigo, P.~Salimov, and E.~Vandomme.
\newblock Some properties of abelian return words.
\newblock {\em J. Integer Seq.}, 16:13.2.5, 2013.

\end{thebibliography}

\end{document}